\documentclass{article}
\usepackage{geometry}
\usepackage[utf8]{inputenc}
\usepackage{amsfonts,amsmath}
\usepackage{amsthm}
\usepackage{mdframed}
\usepackage{lipsum}
\usepackage{hyperref}

\title{Strong Hardness of Approximation for Tree Transversals}
\author{Euiwoong Lee \\ University of Michigan \and Pengxiang Wang \\ University of Michigan}
\date{}

\newtheorem{theorem}{Theorem}

\newtheorem{claim}{Claim}

\newcommand{\eps}{\varepsilon}

\begin{document}

\maketitle

\begin{abstract}
Let $H$ be a fixed graph. The $H$-Transversal problem, given a graph $G$, asks to remove the smallest number of vertices from $G$ so that $G$ does not contain $H$ as a subgraph. While a simple $|V(H)|$-approximation algorithm exists and is believed to be tight for every $2$-vertex-connected $H$, the best hardness of approximation for any tree was $\Omega(\log |V(H)|)$-inapproximability when $H$ is a star. 

In this paper, we identify a natural parameter $\Delta$ for every tree $T$ and show that $T$-Transversal is NP-hard to approximate within a factor $(\Delta - 1 -\eps)$ for an arbitrarily small constant $\eps > 0$. As a corollary, we prove that there exists a tree $T$ such that $T$-Transversal is NP-hard to approximate within a factor $\Omega(|V(T)|)$, exponentially improving the best known hardness of approximation for tree transversals.
\end{abstract}

\section{Introduction}
Let $H = (V(H), E(H))$ be a fixed {\em pattern graph}. The $H$-Transversal problem is a vertex deletion problem whose input is a graph $G = (V(G), E(G))$ and the goal is to compute the smallest set $S \subseteq V(G)$ such that $G \setminus S$ does not have $H$ as a subgraph. 
Note that in this paper, we focus on the notion of {\em subgraphs} instead of {\em induced subgraphs} and {\em (topological) minors}, both of which have been actively studied through the lens of approximation and parameterized algorithms. We refer the reader to recent papers~\cite{ahn2020towards, fomin2020hitting, kim_et_al:LIPIcs.APPROX/RANDOM.2021.7} and a survey~\cite{feldmann2020survey} on these topics. 

$H$-Transversal either captures or is closely related to fundamental optimization problems including Vertex Cover, Dominating Set, Feedback Vertex Set, and Clique Transversal; see~\cite{guruswami2017inapproximability} and references therein. One natural direction is to characterize the complexity and approximability of $H$-Transversal for every $H$. Lund and Yannakakis~\cite{lund1993approximation} gave the complexity classification, proving that whenever $H$ has an edge, $H$-Transversal becomes NP-hard to solve optimally and in fact APX-hard. However, a complete characterization of approximability for $H$-Transversal is not known yet.

When $H$ is a single edge, $H$-Transversal becomes Vertex Cover that has a simple $2$-approximation algorithm,  which is optimal assuming the Unique Games Conjecture (UGC)~\cite{khot2008vertex}. Indeed, for every $H$, there is a simple $|V(H)|$-approximation algorithm for $H$-Transversal by viewing the problem as a special case of $|V(H)|$-Uniform-Hypergraph Vertex Cover; given $G$, consider a hypergraph $H'$ whose vertex set is $V(G)$ and a set of $|V(H)|$ vertices $\{ v_1, \dots, v_{|V(H)|} \}$ forms a hyperedge if and only if the subgraph induced by them has $H$ as a subgraph. Then $S \subseteq V(G)$ is a $H$-transversal in $G$ if and only if it covers every hyperedge of $H'$, so a $|V(H)|$-approximation algorithm for Hypergraph Vertex Cover for $H'$ implies the same approximation factor for $H$-Transversal.

When $H$ is $2$-vertex-connected, it is know that that this simple approximation algorithm is likely to be tight; assuming the UGC, for any constant $\eps > 0$, it is NP-hard to approximate $H$-Transversal within a factor $(|V(H)| - \eps)$~\cite{brustle2021approximation}. (Without the UGC, the factor becomes $(|V(H)| - 1 - \eps)$~\cite{guruswami2017inapproximability}.)

Given the strong hardness of any $2$-vertex-connected $H$, it is natural to study the case when $H$ is a tree. For trees, most known results are algorithmic. When $H$ is a path or a star (i.e., a tree where every vertex except one is a leaf), there exists an $O(\log |V(H)|)$-approximation algorithm~\cite{lee2017partitioning, guruswami2017inapproximability}. Very recently, it was proved that there exists a $(|V(H)|-1/2)$-approximation algorithm for {\em every} tree $H$~\cite{brustle2021approximation}, showing qualitative differences between trees and $2$-vertex-connected graphs. Prior to this work, the largest inapproximability factor for any tree $H$ is $\Omega(\log |V(H)|)$ when $H$ is a star. Given stars and paths are two {\em extreme examples} of trees (e.g., among trees, stars have the smallest diameter and paths have the largest) and they both admit $O(\log |V(H)|)$-approximations, it is natural to suspect that every tree $H$ admits an $O(\log |V(H)|)$-approximation algorithm.

In this paper, we prove that 
surprisingly (at least to the authors),
it is not the case and there exists a tree $T$ such that $T$-Transversal is NP-hard to approximate within a factor $\Omega(|V(T)|)$. 
Given a tree $T$, let $\chi : V(T) \to \{ 0, 1 \}$ to be a proper $2$-coloring of $T$, and 
\[
\Delta(T) := \min_{i \in \{ 0, 1\}} \max_{v \in \chi^{-1}(i)} \deg_T(v).\]
Note that as the $2$-coloring of any tree is unique up to switching two colors, so $\Delta(T)$ does not depend on the choice of $\chi$. Our main theorem is the following hardness for $T$-Transversal. 
\begin{theorem}
Let $T$ be a fixed tree with $\Delta(T) \geq 3$. For any constant $\eps > 0$, it is NP-hard to approximate $T$-Transversal within a factor of $(\Delta(T) - 1 - \eps)$. 
\label{thm:main}
\end{theorem}
In particular, if $T$ is a {\em double star} (i.e., $V(T) = \{ u_1, \dots, u_{k}, v_1, \dots, v_k \}$ and $E(T) = \{ (u_1, v_1) \} \cup (\cup_{i=2}^k \{ (u_1, u_i), (v_1, v_i) \})$ for some integer $k$), then it is hard to approximte $T$-Transversal within a factor $(|V(T)|/2 - 1 - \eps)$ for any $\eps > 0$. 

\subsection{Techniques}
Like the previous strong inapproximability result for $2$-vertex-connected $H$~\cite{guruswami2017inapproximability}, Theorem~\ref{thm:main} starts from the strong hardness of approximation for $k$-Uniform-Hypergraph Vertex Cover ($k$-HVC). The input is a $k$-uniform hypergraph $H = (V(H), E(H))$ where each hyperedge $e \in E(H)$ contains exactly $k$ vertices, and the goal is to choose the smallest subset $S \subseteq V(H)$ that covers (intersects) every hyperedge $e \in E(H)$. \cite{dinur2005new} proved that it is NP-hard to approximate this problem within a factor $(k - 1 - \eps)$ for any $\eps > 0$. 

Let $T$ be a fixed $2$-vertex-connected graph with $k = |V(T)|$. \cite{guruswami2017inapproximability} constructs a reduction from $k$-HVC to $T$-Transversal by {\em directly replacing each hyperedge with a copy of $T$.} Given a hypergraph $H$ for $k$-HVC, it constructs an extended hypergraph $k$-uniform hypergraph $H'$
by letting $V(H') = V(H) \times [B]$ and replacing each hyperedge $(v_1, \dots, v_k)$ of $H$ by $C$ hyperedges of the form $((v_1, i_1), \dots, (v_k, i_k))$ for randomly chosen $i_1, \dots, i_k \in [B]$ for some parameters $B$ and $C$.
The final $G$ for $T$-Transversal is just obtained by letting $V(G) = V(H')$ and each replacing a hyperedge $e = (v_1, \dots, v_k) \in E(H')$ by edges of $T$ between $v_1, \dots, v_k$. Then one can show the optimal $T$-Transversal for $G$ is essentially the same as the optimal vertex cover for $H'$, which is closely related to the optimal vertex cover for $H$. The proof crucially uses the $2$-vertex-connectivity of $T$. 

The key difference in this paper is how we construct $G$ from $H'$. Instead of directly adding a copy of $T$ for each hyperedge of $H'$, we let $G$ be the {\em vertex-hyperedge incidence graph}; $V(G) = V(H') \cup E(H')$ and for $v \in V(H')$ and $e \in E(H')$, the pair $(v, e)$ is an edge in $G$ if and only if $v \in e$. Then $G$ becomes a bipartite graph where the vertices in one side $E(H')$ has degree exactly $k$. 

Suppose $k = \Delta = \Delta(T)$ and $S \subseteq V(H')$ that covers every $e \in E(H')$. Then, in $G \setminus S$, every vertex $e \in E(H')$ has degree at most $\Delta - 1$, which implies that $G \setminus S$ does not contain any copy of $T$; 
when $\chi : V(T) \to \{ 0, 1 \}$ is a $2$-coloring of $T$, since $G \setminus S$ is still bipartite, any injective homomorphism from $T$ to $G \setminus S$ will map the vertices of $T$ of the same color to one side of the bipartition of $G$, but since each color has a vertex with degree at least $\Delta$, the fact that one side of $G \setminus S$ does not contain any vertex of degree at least $\Delta$ implies that such an injective homomorphism cannot exist!

To prove the other direction (i.e., a good $T$-transversal of $G$ implies a good vertex cover of $H$), we use the same technique of carefully constructing $H'$ from $H$ by letting $V(H') = V(H) \times [B]$ and replacing each hyperedge $(v_1, \dots, v_k)$ of $H$ by many hyperedges of the form $((v_1, i_1), \dots, (v_k, i_k))$. Unlike~\cite{guruswami2017inapproximability}, we do not have to use randomness here and simply create every possible hyperedge. The final construction becomes slightly more complicated because we create $C$ copies of the same hyperedge for technical purposes.

\section{Proof of Theorem~\ref{thm:main}}
Fix a tree $T$ with $\Delta := \Delta(T) \geq 3$.
$\Delta$-Uniform-Hypergraph Vertex Cover is the problem whose input is a $\Delta$-uniform hypergraph $H = (V(H), E(H))$ where every hyperedge $e \in E(H)$ contains exactly $\Delta$ vertices, and the goal is to find the smallest vertex cover $S \subseteq V(H)$. A subset $S$ is called a vertex cover if {\em covers} every hyperedge; i.e., every $e \in E(H)$ satisfies $e \cap S \neq \emptyset$. 
Our starting point is the following hardness for $\Delta$-Uniform-Hypergraph Vertex Cover~\cite{dinur2005new}. 

\begin{theorem} [\cite{dinur2005new}]
For any $\Delta \geq 3$ and $\eps' > 0$, given a $\Delta$-uniform hypergraph $H = (V(H), E(H))$, it is NP-hard to distinguish the following two cases:
\begin{itemize}
    \item Completeness: There exists a vertex cover $S \subseteq V(H)$ with $|S| \leq |V(H)| / (\Delta - 1 - \eps')$. 
    \item Soundness: For every vertex cover $S \subseteq V(H)$,  $|S| \geq (1 - \eps') |V(H)|$. 
\end{itemize}
\label{thm:dgkr}
\end{theorem}

We design a reduction from $\Delta$-Uniform-Hypergraph Vertex Cover to $T$-Transversal. Let $H = (V(H), E(H))$ be an instance of $\Delta$-Uniform-Hypergraph Vertex Cover. Let $B$ and $C$ be positive integers that will be fixed later. Given $H$, the reduction outputs a graph $G = (V(G), E(G))$ as an instance of $T$-Transversal as follows. 
\begin{itemize}
    \item We first construct an extended $\Delta$-uniform hypergraph $H' = (V(H'), E(H'))$, where we replace each vertex of $H$ by a {\em cloud} of $B$ vertices and replace each hyperedge $(v_1, \dots, v_\Delta)$ of $H$ by $B^\Delta$ hyperedges of the form $((v_1, i_1), \dots, (v_\Delta, i_\Delta))$ for $i_1, \dots, i_\Delta \in [B]$, 
    and further duplicate each hyperedge $C$ times; the final hyperedges are of the form 
    $((v_1, i_1), \dots, (v_\Delta, i_\Delta))_j$ where $i_1, \dots, i_\Delta \in [B]$ and $j \in [C]$. Formally, 
    \begin{itemize}
        \item $V(H') = V(H) \times [B]$.
        \item $E(H') = \{ ((v_1, i_1), \dots, (v_\Delta, i_\Delta))_j : (v_1, \dots, v_k) \in E(H)\mbox{ and } i_1, \dots, i_\Delta \in [B], j \in [C] \}$. Note that $|E(H')| = |E(H)|\cdot B^{\Delta}\cdot C$.
    \end{itemize}
    
    \item Let $G = (V(G), E(G))$ be the {\em vertex-hyperedge incidence graph} of $H'$. Formally,
    \begin{itemize}
        \item $V(G) = V(H') \cup E(H')$.
        \item $E(G) = \{ (v, e) : v \in V(H'), e \in E(H')\mbox{ and }v \in e \}$. 
    \end{itemize}
\end{itemize}
Note that $G$ is a bipartite graph. 

\paragraph{Completeness.}
Suppose that there exists $S \subseteq V(H)$ such that $|S| \leq |V(H)|/(\Delta - 1 - \eps')$ and it covers every hyperedge of $H$; i.e., for every $e \in E(H)$, $S \cap e \neq \emptyset$. Let $S' = S \times [B] \subseteq V(H')$ such that $|S'| \leq B|V(H)|/(\Delta - 1 - \eps')$. It is simple to verify that $S'$ covers every hyperedge of $H'$ as well; for every hyperedge $e' = ((v_1, i_1), \dots, (v_\Delta, i_\Delta))_{\ell}$ of $H'$, $(v_1, \dots, v_\Delta)$ is a hyperedge of $H$, which implies that $S$ contains some $v_j$ for $j \in [\Delta]$ and $S'$ contains $(v_j, i_j)$. 

We would like to prove that $S'$, as a subset of $V(G)$, is a valid $T$-transversal; it covers every copy of $T$ in $G$. This follows from the fact that after deleting $S'$ from $G$, every vertex $e \in E(H')$ has a degree at most $\Delta - 1$ in $G \setminus S'$; it has degree exactly $\Delta$ in $G$, but since $S'$ is a vertex cover for $H'$, there exists $v \in S'$ such that $(e, v) \in E(G)$.
Then $G \setminus S'$ is a bipartite graph where the maximum degree on one side is at most $\Delta - 1$. Since $T$ is a bipartite graph where both sides have a vertex of degree at least $\Delta$, $T$ cannot be a subgraph of $G \setminus S'$.

\paragraph{Soundness.}
Suppose that for every vertex cover $S \subseteq V(H)$,  $|S| \geq (1 - \eps') |V(H)|$. Let $R \subseteq V(G)$ be an optimal $T$-transversal of $G$. Our choice of $B$ and $C$ will satisfy
\begin{equation}
C > 2|V(H')| = 2|V(H)| \cdot B. 
\label{eq:cond1}
\end{equation}
Given this choice, we can prove that $R$ only contains vertices from $V(H')$, not $E(H')$. 
\begin{claim}
$R \subseteq V(H')$. 
\end{claim}
\begin{proof}
If there exists $((v_1, i_1), \dots, (v_\Delta, i_\Delta)) \in (V(H) \times [B])^{\Delta}$ such that 
$|R \cap \{ ((v_1, i_1), \dots, (v_\Delta, i_\Delta))_j : j \in [C] \}| > |V(H')|$, it violates the optimality of $R$; just taking all vertices in $V(H')$ is a cheaper $T$-transversal.
Therefore, we assume that for any $((v_1, i_1), \dots, (v_\Delta, i_\Delta))$, we have 
$|R \cap \{ ((v_1, i_1), \dots, (v_\Delta, i_\Delta))_j : j \in [C] \}| \leq |V(H')| < C / 2$. 

We now claim that $R \cap V(H')$ is a $T$-transversal. Consider any set of $|V(T)|$ vertices $I = \{ v_1, \dots, v_{p} \} \cup \{ (e_1)_{j_1}, \dots, (e_q)_{j_q} \}$ of $G$ whose induced subgraph $G_I$ contains $T$ as a subgraph, where $\{ v_1, \dots, v_p \} \subseteq V(H')$ and 
$\{ (e_1)_{j_1}, \dots, (e_q)_{j_q} \} \subseteq E(H')$ (e.g., for each $\ell \in [q]$, $e_{\ell} \in (V(H) \times [B])^{\Delta}$ and $j_{\ell} \in [C]$). For each $e_{\ell}$, among $C$ identical copies from $\{(e_{\ell})_{j_{\ell}}\}_{j_{\ell} \in [C]}$, $R$ contains at less than $C/2$ copies. Therefore, one can find $j'_1, \dots, j'_q$ such that none of $(e_1)_{j'_1}, \dots, (e_q)_{j'_q}$ is contained in $R$. For every $\ell \in [q]$, $(e_{\ell})_{j_{\ell}}$ and $(e_{\ell})_{j'_{\ell}}$ have the exactly the same of neighbors, so one can conclude that $I' = \{ v_1, \dots, v_{p} \} \cup \{ (e_1)_{j'_1}, \dots, (e_q)_{j'_q} \}$ also contains $T$ in its induced subgraph. However, by construction 
$R$ contains none of $(e_1)_{j'_1}, \dots, (e_q)_{j'_q}$, which means that $R$ contains at least one vertex from $\{ v_1, \dots, v_{p} \}$. This implies that $R \cap V(H')$ also contains at least one vertex from $\{ v_1, \dots, v_{p} \}$, which implies that $R \cap V(H')$ is a $T$-transversal.

By optimailty of $R$, we have $R \cap V(H') = R$, which implies that $R \subseteq V(H')$. 
\end{proof}

Let $k = |V(T)|$ and $w$ be a constant that will be fixed later only depending on $k$, and for each $v \in V(H)$, say $v$ is {\em occupied} if $|R \cap (\{ v \} \times [B])| \geq B - w$, and {\em free} otherwise. 
For $e = (v_1, \dots, v_{\Delta}) \in V(H)$, call $e$ {\em free} if all $v_1, \dots, v_{\Delta}$ are free. 

\begin{claim}
No $e \in V(H)$ is free. 
\end{claim}
\begin{proof}
Assume towards contradiction that $e = (v_1, \dots, v_{\Delta}) \in V(H)$ is free; all $v_1, \dots, v_{\Delta}$ are free. We will show that $R$ is not a $T$-transversal.

Without loss of generality, after suitable permutations of vertices, assume that for each $\ell \in [\Delta]$, none of $(v_{\ell}, 1), \dots, (v_{\ell}, w)$ is in $R$. We will find a large tree $T'$ in the subgraph of $G$ induced by $V' \cup E'$ where 
$V' = (\cup_{\ell \in [\Delta]} (\{ v_{\ell} \} \times [w]))$ and $E' = \{ ((v_1, i_1), \dots, (v_{\Delta}, i_{\Delta}))_1 : i_1, \dots, i_{\Delta} \in [w] \}$. The tree $T'$ has height is $2k - 1$ and it has $2k$ levels from $0$ to $2k - 1$. Each even level contains a node from $E'$ and each odd level contains a node from $V'$; furthermore, each odd-level node has {\em type $\ell$} when it contains a node from $(v_{\ell} \times [w])$. Each even-level node of $T'$ will have degree $\Delta$ and each odd-level internal node (i.e., at level $1, 3, \dots, 2k - 3$) of $T'$ will have degree $k$. Let the root node be $((v_1, 1), \dots, (v_{\Delta}, 1))_1$ and its $\Delta$ children be $(v_1, 1), \dots, (v_{\Delta}, 1)$. 
The rest of $T'$ is constructed by the following procedure run for each odd-level node. 

\begin{itemize}
    \item For each odd-level node $(v_{\ell}, i_{\ell})$ of type $\ell$: 
    \item If the current level is already $2k - 1$, return.
    \item Otherwise, for each $\ell' \in [\Delta] \setminus {\ell}$, choose $k - 1$ new vertices from $\{ v_{\ell'} \} \times [w]$ that have not been chosen during the construction of $T'$. Call them $(v_{\ell'}, i'_{\ell', 1}), \dots, (v_{\ell'}, i'_{\ell', k - 1})$. 
    \begin{itemize}
        \item Since $T'$ has at most $(k\Delta)^{k}$ internal nodes, by ensuring 
        \begin{equation}
            w > k^{3k} \geq (k\Delta)^{k + 1},
            \label{eq:cond2}
        \end{equation}
        one can ensure that this process can be done for every odd-level internal node. 
    \end{itemize}
    \item For each $r = 1, \dots, k - 1$,
    \begin{itemize}
        \item Create a (even-level) child $((v_1, i'_{1,r}), \dots, (v_{\ell-1}, i'_{\ell - 1,r}), (v_{\ell}, i_{\ell}),
        (v_{\ell+1}, i'_{\ell + 1,r}),
        \dots, (v_\Delta, i'_{\Delta,r}))_1$.
        \begin{itemize}
            \item Its $\Delta - 1$ (odd-level) children will be $(v_1, i'_{1,r}), \dots, (v_{\ell-1}, i'_{\ell - 1,r}), 
            (v_{\ell+1}, i'_{\ell + 1,r}), \dots, (v_\Delta, i'_{\Delta,r})$.
        \end{itemize}
    \end{itemize}
\end{itemize}

Therefore, one can conclude that a desired $T'$ can be found from $V' \cup E'$. 
Since $T'$ has height $2k - 1$ and every even-level node has degree exactly $\Delta$ and every odd-level internal node has degree exactly $k$, we claim that $T'$ contains a copy of $T$. If $\chi : V(T) \to \{ 0, 1 \}$ is a $2$-coloring of $T$ such that $\max_{V \in \chi^{-1}(0)} \deg_T(v) = \Delta$, mapping any fixed node $v \in \chi^{-1}(0)$ to the root of $T'$ and arbitrarily extending the mapping along the edges of $T$ will give a injective homomorphism from $T$ to $T'$; every node in $\chi^{-1}(0)$ will be mapped to even-level nodes of $T'$ and every node in $\chi^{-1}(1)$ will be mapped to odd-level nodes of $T'$, both of which have enough degrees (i.e., $\Delta$ for even levels, $k$ for odd levels) for further extension. 

Finally, note that since $R \cap (V' \cup E') = \emptyset$, $R$ does not intersect $T'$. Since $T'$ contains a copy of $T$, it contradicts that $R'$ is a $T$-transversal and finishes the proof. 
\end{proof}

Since no $e \in V(H)$ is free, it implies that the set of occupied verties is a valid hypergraph transversal in $H$, which implies that $|R| \geq (1 - \eps')|V(H)|(B - w)$. By setting $w$ be a constant greater than $k^{3k}$, $B = \omega(w)$, and $C > 2|V(H)|B$ satisfies all the previous conditions (\eqref{eq:cond1} and~\eqref{eq:cond2}) while ensuring that $|R| \geq (1 - \eps' - o(1))|V(H)|B$. The multiplicative gap between the sizes of the optimal $T$-transversal between the completeness case and the soundness case is at least $(\Delta - 1)(1 - O(\eps') - o(1))$. 

\bibliographystyle{alpha}
\bibliography{reference}

\end{document}